\pdfoutput=1

\documentclass[letterpaper, 10 pt, conference]{ieeeconf}  

\IEEEoverridecommandlockouts                              
\usepackage{amsmath} 
\usepackage{amssymb}  
\usepackage{graphicx}
\usepackage{subcaption}
\usepackage{eso-pic}
\AddToShipoutPictureFG*{%
  \AtPageUpperLeft{%
    \put(\LenToUnit{0.75in},\LenToUnit{-0.45in}){%
      \parbox[t]{\dimexpr\paperwidth-1.5in\relax}{\centering\scriptsize
        \copyright\ 2026 IEEE. Personal use of this material is permitted. Permission from IEEE must be obtained for all other uses, in any current or future media, including reprinting/republishing this material for advertising or promotional purposes, creating new collective works, for resale or redistribution to servers or lists, or reuse of any copyrighted component of this work in other works.}}}}

\newtheorem{definition}{Definition}
\makeatletter
\def\@opargbegintheorem#1#2#3{\trivlist
   \item[\hskip \labelsep{\itshape{#1}\ \itshape{#2}}]{(\itshape{#3}):}\ }
\makeatother

\newtheorem{assumption}{Assumption}

\newtheorem{theorem}{Theorem}

\title{\LARGE \bf
Strategic Arbitrage in Electricity Markets: A Multi-Agent Policy Optimization Framework
}

\author{Eli Brock, Srikar Babu Gadipudi, Somayeh Sojoudi, and Javad Lavaei
\thanks{The authors are with the University of California, Berkeley.}%
\thanks{This work used Jetstream2 at Indiana University through allocations ELE250009 and CIS251274 from the Advanced Cyberinfrastructure Coordination Ecosystem: Services \& Support (ACCESS) program, which is supported by U.S. National Science Foundation grants \#2138259, \#2138286, \#2138307, \#2137603, and \#2138296 \cite{boerner2023access}.}%
\thanks{This work was supported by the U. S. Army Research Laboratory and the U. S. Army Research Office under Grant W911NF2010219, Office of Naval Research, and NSF.}
}

\begin{document}

\maketitle
\thispagestyle{empty}
\pagestyle{empty}
\begin{abstract}

Grid-scale storage and load aggregators are growing quickly in wholesale electricity markets. If highly consolidated, these new participants may distort the competitive market equilibrium through strategic withholding behavior. Traditional market power mitigation techniques are ill-suited to arbitrage-based participants. We provide a multi-agent policy optimization model to study game-theoretic equilibria for real-time markets with consolidated storage and/or aggregators. The model is designed to facilitate natural first-order equilibrium-seeking update rules. We analyze the theoretical properties of different supply functions and their implications on the convergence of the gradient dynamics. Finally, we train equilibrium policies for a modified version of the IEEE RTS-GMLC benchmark across different levels of storage fleet consolidation, to demonstrate the potential of the framework as a tool to characterize and quantify strategic arbitrage.

\end{abstract}

\section{INTRODUCTION} \label{sec:intro}

Modern electricity markets with large penetrations of variable renewable energy (VRE) typically also feature significant grid-scale storage. Like large conventional generator fleets, large storage operators might exercise \textit{market power} by withholding some arbitrage potential to manipulate prices. However, the cost of storing energy at a given time is characterized by future arbitrage potential, rather than publicly known fuel price, so Independent System Operators (ISOs) lack a clear falsifiable mechanism for recognizing withholding behavior in storage units as they do for conventional generators \cite{anunrojwong2024battery}. \cite{ma2025comparative} concluded that storage operators in California may already be acting as price-makers. A similar problem might soon apply to load aggregators, recently permitted to participate in wholesale markets by FERC Order 2222. A small number of sizable aggregators may prove even more difficult to regulate than consolidated grid-scale storage due to the operational complexity of coordinating millions of devices on the edge of the grid.

A recent stream of literature has addressed strategic withholding by grid-scale storage operators. \cite{anunrojwong2024battery} provides a rigorous closed-form characterization of withholding behavior on a stylized model. \cite{wu2024market} develops a more realistic framework to simulate and identify withholding behaviors. \cite{qi2025locational} proposes a method for a system operator to set bounds on energy storage bids to regulate market power. \cite{contreras2017participation} models a price-making load aggregator. Our work is distinct from existing literature in that we model the real-time market as a stochastic game \cite{shapley1953stochastic} where operators learn closed-loop policies as opposed to open-loop bid schedules over some fixed horizon. This Markovian game-theoretic framework is a more accurate representation of the decision problem faced by storage operators, allows for a high-fidelity model of the market-clearing process, and facilitates a first-class treatment of the uncertainty inherent in power grids. To our knowledge, this is the first work to frame arbitrage-based withholding in electricity markets as a multi-agent policy optimization game.

Section \ref{sec:game-theory} introduces the necessary game-theoretic concepts. Section \ref{sec:formulation} formulates an electricity market with self-dispatching strategic agents as a stochastic game, and Section \ref{sec:analysis} analyzes two smooth instances of the general game. Section \ref{sec:case-study} conducts a detailed experiment on the IEEE RTS-GMLC benchmark case \cite{barrows2019ieee}: stochastic gradient play converges reliably to game-theoretic equilibria, and these equilibria exhibit the expected market-power effect, with consolidation raising storage revenue at the expense of social welfare. The case study indicates that the proposed framework is a promising tool to characterize the economic implications of large, consolidated grid-scale storage operators and load aggregators.

\subsubsection*{Notation}
$\text{LSE}_\zeta(x) = \zeta\log\left(\sum_i\exp\left(\frac{1}{\zeta}x_i\right)\right)$ is the LogSumExp approximation of the maximum operator with smoothing parameter $\zeta > 0$. $\text{LSE}_\zeta^0(x)\equiv\text{LSE}_\zeta(0,x)$. $\mathcal{P}(\bullet)$ is the set of Borel distributions over the set $\bullet$. $N(x,\sigma^2)$ is the normal distribution with mean $x$ and variance $\sigma^2$. $2^\bullet$ denotes the power set of the set $\bullet$. $\nabla_\bullet$ and $\frac{\partial}{\partial \bullet}$ denote the gradient and Jacobian with respect to $\bullet$, respectively. $\Re$ is the real-part operator and $\text{spec}$ is the spectrum operator.

\section{GAME THEORY PRELIMINARIES} \label{sec:game-theory}

\subsection{Continuous games} \label{sec:continuous-games}
An unconstrained continuous game with set of players $\mathcal{N}$ is specified by a tuple of \textit{payoff functions}, $(J_i)_{i\in\mathcal{N}}$ where $J_i: \mathbb{R}^{n_i}\rightarrow\mathbb{R}$ for all $i\in\mathcal{N}$ and $n_i$ is the dimension of player $i$'s strategy space. A strategy profile $\theta$ is a \textit{local Nash equilibrium} if no player can improve their payoff by a local deviation; see \cite{mazumdar2020gradient} for a precise definition. Smooth continuous games, where each $J_i$ is smooth, admit an equivalent first and second-order characterization:
\begin{definition}[Differential Nash Equilibrium \cite{mazumdar2020gradient}] \label{def:differential_equilibrium}
    A policy profile $\theta$ is called a \textit{differential Nash equilibrium} of a smooth continuous game $(J_i)_{i\in\mathcal{N}}$ if $\nabla_{\theta_i}J_i(\theta) = 0$ and $\nabla_{\theta_i}^2J_i(\theta)\prec0$ for all $i\in\mathcal{N}$.
\end{definition}

Standard equilibrium-seeking algorithms for smooth continuous games feature update rules designed to mimic the \textit{gradient flow}
\begin{equation} \label{eq:gradient-flow}
    \dot{\theta}_i = \Gamma_i(\theta) \equiv \nabla_{\theta_i} J_i(\theta) \quad \forall i \in \mathcal{N}.
\end{equation}
\begin{definition}[Locally Asymptotically Stable Equilibrium \cite{mazumdar2020gradient}] \label{def:lase}
    A policy profile $\theta$ is called a \textit{locally asymptotically stable equilibrium} (LASE) of the smooth continuous game $(J_i)_{i\in\mathcal{N}}$ if $\nabla_{\theta_i}J_i(\theta) = 0$ for all $i\in\mathcal{N}$ and the \textit{game Jacobian} $\frac{\partial}{\partial\theta}\Gamma(\theta)$ of the vector field in \eqref{eq:gradient-flow} is Hurwitz, that is, $\Re(\lambda) < 0$ for all $\lambda\in\text{spec}\left(\frac{\partial}{\partial\theta}\Gamma(\theta)\right)$.
\end{definition}

Every Nash equilibrium is either an LASE or a strict saddle point \cite{mazumdar2020gradient}. For general-sum games, however, \eqref{eq:gradient-flow} is a heuristic: its dynamics may exhibit limit cycles or divergence \cite{fudenberg1991game}, and LASEs that are not Nash equilibria may exist \cite{mazumdar2020gradient}. Convergence guarantees are restored for the class of \textit{potential games}:
\begin{definition}[Potential Game \cite{monderer1996potential}] \label{def:potential_game}
    A smooth continuous game is called a \textit{potential game} if there exists a potential function $\Phi$ such that
    \begin{equation*}
        \nabla_{\theta_i}\Phi(\theta) = \nabla_{\theta_i}J_i(\theta)
    \end{equation*}
    for all $i \in \mathcal{N}$.
\end{definition}

For potential games, \eqref{eq:gradient-flow} reduces to cooperative gradient ascent on $\Phi$. Every stationary point is therefore a Nash equilibrium, and \eqref{eq:gradient-flow} converges under mild assumptions \cite{swenson2018best}.

\subsection{Stochastic games} \label{sec:stochastic-games}
A stochastic game \cite{shapley1953stochastic} is a tuple $S = (\mathcal{S}, (\mathcal{A}_i)_{i \in \mathcal{N}}, (R_i)_{i \in \mathcal{N}}, f, P_\omega, P_0, \gamma)$ where $\mathcal{N}$ is a finite set of agents, $\mathcal{S}$ is the state space, $\mathcal{A}_i: \mathcal{S} \to 2^{\mathcal{A}_i^\cup}$ maps each state to agent $i$'s feasible action set (where $\mathcal{A}_i^\cup=\bigcup_{s\in\mathcal{S}}\mathcal{A}_i(s)$), $R_i: \mathcal{K} \to \mathbb{R}$ is agent $i$'s reward function with $\mathcal{K} \equiv \{(s,a): s \in \mathcal{S},\, a_i \in \mathcal{A}_i(s)\quad\forall i\in\mathcal{N}\}$, $f: \mathcal{K} \times \Omega \to \mathcal{S}$ is the transition function, $P_0\in\mathcal{P}(\mathcal{S})$ and $P_\omega\in\mathcal{P}(\Omega)$ are the initial state and noise densities, and $\gamma \in [0,1)$ is the discount factor.

Each agent selects actions via a stationary parameterized policy $\mu_i: \mathcal{S} \times \mathbb{R}^{n_i} \to \mathcal{A}_i^\cup$ satisfying $\mu_i(s|\theta_i) \in \mathcal{A}_i(s)$ for all $s \in \mathcal{S}, \theta_i\in\mathbb{R}^{n_i}$. A stochastic game $S$ together with parameterizations $\Pi = (\mu_i)_{i \in \mathcal{N}}$ induces a continuous game with payoff functions
\begin{equation} \label{eq:expectea_return}
    J_i(\theta) \equiv \mathbb{E}\left[\sum_{t=0}^\infty\gamma^tR_i(s_t,a_t)\right],
\end{equation}
where $s_0 \sim P_0$, $\omega_t \sim P_\omega$, $a_{i,t} = \mu_i(s_t|\theta_i)$, and $s_{t+1} = f(s_t, a_t, \omega_t)$. We call $(S, \Pi)$ a \emph{parameterized stochastic game} (PSG). Practical potential PSGs typically belong to the class of \textit{Markov potential games}:
\begin{definition}[Markov Potential Game \cite{leonardos2021global}] \label{def:mpg}
    A smooth PSG $(S, \Pi)$ is a \emph{Markov potential game} (MPG) if there exists a function $\phi:\mathcal{K}\to\mathbb{R}$ such that $(S, \Pi)$ is a potential game with $\Phi(\theta) = \mathbb{E}\left[\sum_{t=0}^{\infty} \gamma^t \phi(s_t, a_t)\right]$.
\end{definition}

\subsection{Learning procedure: stochastic gradient play}
In the stochastic game setting, the returns $J_i(\theta)$ and their gradients cannot generally be computed in closed form. Given oracle access to $(S, \Pi)$, one can estimate returns by sampling an initial state from $P_0$ and a noise trajectory i.i.d.\ from $P_\omega$:
\begin{equation*}
    \hat{J}_{i,T}\left(\theta, s_0, \omega_{0:T-1}\right) = \sum_{t=0}^T R_i(s_t, a_t),
\end{equation*}
where the state and action trajectories follow from the parameterized dynamics. We say a PSG is \textit{per-sample smooth} if $\hat{J}_{i,T}(\bullet,s_0,\omega_{0:T-1})$ is smooth for all $s_0\in\mathcal{S}$ and $\omega_{0:T-1}\in\Omega^T$. Per-sample smooth PSGs therefore allow for \textit{stochastic gradient play} (SGP), the core learning procedure of this work, which approximates \eqref{eq:gradient-flow} by updates of the form
\begin{equation} \label{eq:sgp} \tag{SGP}
    \theta_i^{(k+1)} = \theta_i^{(k)} + \frac{\eta^{(k)}}{B}\sum_{j=1}^B\nabla_{\theta_i}\hat{J}_{i,T^{(j)}}(\theta^{(k)}, s_0^{(j)}, \omega^{(j)}),
\end{equation}
where $B\in\mathbb{Z}_{++}$ is a batch size and $s_0^{(j)}$, $\omega^{(j)}$, and $T^{(j)}$ are sampled i.i.d. from the initial state density, the noise density, and a geometric distribution with success probability $\gamma$, respectively, yielding an unbiased estimator. Since we seek only to compute equilibria rather than to simulate learning (the ``computational'' agenda of \cite{shoham2007if}), we compute \eqref{eq:sgp} directly by backpropagating through rollouts \cite{heess2015learning}. This approach allows for simpler, lower-variance gradient estimation than model-free policy gradient methods and is compatible with deterministic policies. 

\section{FORMULATION} \label{sec:formulation}
We propose a PSG model of a real-time wholesale electricity market with a fleet of conventional generators and a number of flexible devices, such as grid-scale storage units or consumer-owned distributed energy resources. Management of the devices is divided between a small number of agents, each of whom may have an appreciable influence on the market price.

We assume the agents participate via \textit{self-schedule} bids, meaning that they inform the system operator of their level of consumption or production, as opposed to \textit{economic bids}, which consist of quantity/price pairs. We also neglect the strategic implications of the agents' forward contracts, such as power purchase agreements or a position in a voluntary day-ahead market. For simplicity, we assume the electric transmission network connecting the agents is unconstrained. Since ISOs typically clear markets using DC approximate power flow equations that neglect losses, an unconstrained network is equivalent to a pooled single-node market. This assumption precludes locational market power arising from congestion; the equilibria studied here therefore characterize system-wide rather than locational withholding.

We reuse the notation of Section \ref{sec:stochastic-games} for the corresponding objects of the specific game developed here. We also construct closely related objects with restricted function signatures, marked with a tilde $\tilde{\bullet}$; the two are related formally in \eqref{eq:psg-construction}--\eqref{eq:policies}.

As in Section \ref{sec:game-theory}, $\mathcal{N}$ denotes the set of self-dispatching strategic agents. Each agent $i\in\mathcal{N}$ manages a finite set of devices $\mathcal{D}_i$, where $\mathcal{D}_i\cap\mathcal{D}_j = \emptyset$ for all $i\neq j$ without loss of generality. Denote the set of all devices $\mathcal{D}\equiv\bigcup_{i\in\mathcal{N}}\mathcal{D}_i$. Each device $i\in\mathcal{D}$ has a private state space $\mathcal{S}_i$, which may represent, for example, the state-of-charge range of a battery storage unit or the feasible indoor air temperatures of an HVAC system managed by a load aggregator. $\mathcal{S}_0$ is the public market state space capturing any exogenous quantities relevant to the grid, including seasonality, weather, or abstractions thereof. The market state evolves according to an action-independent transition function $\tilde{f}_0:\mathcal{S}_0\times\Omega \to \mathcal{S}_0$.

For each device $i\in\mathcal{D}$, let $\tilde{\mathcal{A}}_i:\mathcal{S}_i\to2^{\tilde{\mathcal{A}}_i^\cup}$ denote its action space mapping, where $\tilde{\mathcal{A}}_i^\cup=\bigcup_{s\in\mathcal{S}_i}\tilde{\mathcal{A}}_i(s)$ and $\tilde{\mathcal{A}}^\cup\equiv\prod_{i\in\mathcal{D}}\tilde{\mathcal{A}}_i^\cup$. Device $i$'s action is its scalar energy consumption, that is $\tilde{\mathcal{A}}_i^\cup\subset\mathbb{R}$. Each device $i\in\mathcal{D}$ has an associated policy $\tilde{\mu}_i:\mathcal{S}_0\times\mathcal{S}_i\times\mathbb{R}^{n_i}\to\tilde{\mathcal{A}}_i^\cup$
\begin{equation*}
    \tilde{\mu}_i\left(s_0,s_i|\theta_i\right)\in\tilde{\mathcal{A}}_i(s_i) \quad
    \forall s_0 \in \mathcal{S}_0,
    s_i \in \mathcal{S}_i,
    \theta_i \in \mathbb{R}^{n_i}.
\end{equation*}
Device $i\in\mathcal{D}$'s state evolves according to a dynamics function $\tilde{f}_i:\mathcal{S}_0\times\mathcal{K}_i \times \Omega \to \mathcal{S}_i$, where $\mathcal{K}_i \equiv \{(x,d):x\in\mathcal{S}_i,d\in\tilde{\mathcal{A}}_i(x)\}$. Notice that both $\tilde{\mu}_i$ and $\tilde{f}_i$ are conditioned only on the public market state and device $i$'s state; neither depends on the private states of other agents.

The agents are compensated for their net production (equivalently, charged for their net consumption) according to the market-clearing price, resulting in payoff function $\tilde{R}_i:\mathcal{S}_0\times\tilde{\mathcal{A}}^\cup\to\mathbb{R}$ for each agent $i\in\mathcal{N}$ given by 
\begin{equation} \label{eq:payout}
    \tilde{R}_i\left(s_0,a\right) = -\underbrace{C'}_{\substack{\text{price} \\ \text{function}}}\left(\underbrace{m(s_0)}_{\substack{\text{baseline} \\ \text{demand}}} + \sum_{j\in\mathcal{D}}a_j\right)\sum_{j\in\mathcal{D}_i} a_j.
\end{equation}
$m:\mathcal{S}_0\to\mathbb{R}$ maps the market state to the current \textit{baseline demand}, or the inelastic component of demand purchased by load-serving entities. $C:\mathbb{R}\to\mathbb{R}$ maps a power demand to the total cost incurred by the system to meet that demand over a single unit of time. As in typical wholesale market-clearing processes, the system price is determined by the \textit{marginal cost}, that is, the price function is the derivative of the system cost function. Notice that $C$ is stationary, implicitly assuming that any VRE generation is modeled as a negative contribution to the net demand argument, rather than as a time-varying part of the expression of $C$.

Per Section \ref{sec:game-theory}, we can now formally construct a stochastic game with
\begin{subequations} \label{eq:psg-construction}
\begin{align}
    \mathcal{S} & = \mathcal{S}_0 \times \prod_{i\in\mathcal{D}}\mathcal{S}_i \\
    \mathcal{A}_i(s) & = \prod_{j\in\mathcal{D}_i}\tilde{\mathcal{A}}_j(s_j) \quad \forall i\in\mathcal{N} \\
    R_i(s, a) & = \tilde{R}_i\left(s_0,a\right) \quad \forall i\in\mathcal{N} \\
    f_0\left(s, a, \omega\right) & = \tilde{f}_0(s_0,\omega) \label{eq:public-transition-function} \\
    f_i\left(s, a, \omega\right) & = \tilde{f}_i\left(s_0,s_i,a_i,\omega\right)  \quad \forall i\in\mathcal{D}, \label{eq:private-transition-function}
\end{align}
\end{subequations}

where $f_0$ and $f_i$ are the components of $f$ that map to the public market state and device $i$'s private state, respectively. We can also construct the policy parameterization $\Pi=(\mu_i)_{i\in\mathcal{N}}$ as
\begin{equation}
    \mu_{i,j}\left(s | \theta_{\mathcal{D}_i}\right) = \tilde{\mu}_j\left(s_0,s_j|\theta_j\right) \quad \forall j \in \mathcal{D}_i, \label{eq:policies}
\end{equation}
where $\mu_{i,j}$ is the component of agent $i$'s policy that maps to device $j\in\mathcal{D}_i$'s action.

Recall from Section \ref{sec:game-theory} that $(S,\Pi)$ is the PSG induced by $S$ and $\Pi$. Under Assumption \ref{assumption:smoothness}, one can easily verify that $(S,\Pi)$ is per-sample smooth.
\begin{assumption} \label{assumption:smoothness}
    $C$, $m$, $\tilde{f}_0$, as well as $\tilde{f}_i$ and $\tilde{\mu}_i$ for all $i\in\mathcal{D}$, are smooth functions.
\end{assumption}

\section{ANALYSIS} \label{sec:analysis}
\subsection{Potential game structure}
Theorem \ref{thm:potential-structure} establishes that determining whether $(S,\Pi)$ is a potential game reduces to determining whether the stagewise payoffs are potential in the agent demands $a$. Theorem \ref{thm:potential-structure} specializes \cite[Theorem 1]{brock2025coordinating} to the smooth PSG paradigm adopted in this work to facilitate a straightforward, directly verifiable proof.
\begin{theorem} \label{thm:potential-structure}
    If there exists a function $\tilde{\phi}:\mathcal{S}_0\times\tilde{\mathcal{A}}^\cup\to\mathbb{R}$ such that
    \begin{equation} \label{eq:general-stagewise-potential-function}
        \frac{\partial}{\partial a_j}\tilde{\phi}\left(s_0,a\right) = \frac{\partial}{\partial a_j}\tilde{R}_i\left(s_0,a\right)
    \end{equation}
    for all $i\in\mathcal{N}$ and $j\in\mathcal{D}_i$, then $(S,\Pi)$ is an MPG with
    \begin{equation} \label{eq:stagewise-potential}
        \phi\left(s, a\right) = \tilde{\phi}\left(s_0,a\right).
    \end{equation}
\end{theorem}
\begin{proof} 
    Throughout this proof, we use the shorthand $s_{t+1} = f(s_t,a_t,\omega_t)$ and $a_{j,t}=\tilde{\mu}_j(s_{0,t},s_{j,t}|\theta_j)$ for all $j\in\mathcal{D}$ and $t\in\mathbb{Z}_+$. For all $i\in\mathcal{N}$, $j\in\mathcal{D}_i$, and $t\in\mathbb{Z}_+$, we have
    \begin{subequations} \label{eq:stagewise-equivalence}
    \begin{align}
        \nabla_{\theta_j}\tilde{R}_i(s_{0,t},a_t) & = \frac{\partial s_{0,t}}{\partial \theta_j}^\top\nabla_{s_0}\tilde{R}_i\left(s_{0,t},a_t\right) \\
        & + \frac{\partial a_t}{\partial \theta_j}^\top\nabla_a\tilde{R}_i\left(s_{0,t},a_t\right) \nonumber \\
        & = \frac{\partial}{\partial a_j}\tilde{R}_i\left(s_{0,t},a_t\right)\nabla_{\theta_j}a_{j,t} \label{eq:independence-step-one} \\
        & = \frac{\partial}{\partial a_j}\tilde{\phi}\left(s_{0,t},a_t\right)\nabla_{\theta_j}a_{j,t} \label{eq:stagewise-potential-step} \\
        & = \nabla_{\theta_j}\tilde{\phi}\left(s_{0,t},a_t\right). \label{eq:independence-step-two}
    \end{align}
    \end{subequations}
    The step \eqref{eq:independence-step-one} uses the fact that the trajectories $s_{0,0:\infty}$ and $a_{k,0:\infty}$, $k \neq j$ evolve independently of $\theta_j$, which is true by construction of \eqref{eq:public-transition-function}, \eqref{eq:private-transition-function}, and \eqref{eq:policies}. \eqref{eq:independence-step-two} re-applies this fact in the reverse direction. Step \eqref{eq:stagewise-potential-step} applies \eqref{eq:general-stagewise-potential-function}. \eqref{eq:stagewise-equivalence} implies $\nabla_{\theta_j}J_i\left(\theta\right) = \mathbb{E}\left[\sum_{t=0}^\infty\gamma^t\nabla_{\theta_j}\tilde{R}_i\left(s_{0,t},a_t\right)\right] = \mathbb{E}\left[\sum_{t=0}^\infty\gamma^t\nabla_{\theta_j}\tilde{\phi}\left(s_{0,t},a_t\right)\right] = \nabla_{\theta_j}\Phi\left(\theta\right)$, where the expectation is taken with respect to $s_0 \sim P_0$ and $\omega_t \sim P_\omega$. Since agent $i$'s parameters are $\theta_{\mathcal{D}_i}$, this implies $\nabla_{\theta_{\mathcal{D}_i}}J_i(\theta)=\nabla_{\theta_{\mathcal{D}_i}}\Phi(\theta)$ for all $i\in\mathcal{N}$, satisfying Definition \ref{def:mpg} with $\phi$ as in \eqref{eq:stagewise-potential}.
\end{proof}

Intuitively, Theorem \ref{thm:potential-structure} says that the game as formulated in Section \ref{sec:formulation} has the \textit{structural} properties necessary to be an MPG and that the existence of a \textit{stagewise} potential function $\tilde{\phi}$ is the determining factor.

\subsection{Cost functions}
In practice, ISOs construct the cost function by solving the \textit{economic dispatch problem}
\begin{subequations} \label{eq:economic-dispatch}
\begin{align}
    C^*(d) = \min_p \quad & \sum_{i\in\mathcal{G}} c_ip_i \\
    \text{s.t.} & \sum_{i\in\mathcal{G}} p_i \geq d \label{eq:power_balance} \\
    & 0 \leq p_i \leq \overline{p}_i & \forall i \in \mathcal{G}, \label{eq:gen_capacity}
\end{align}
\end{subequations}
where $\mathcal{G}$ is the set of generators and generator $i\in\mathcal{G}$'s bid is its fuel cost $c_i\geq0$ and its maximum power output $\overline{p}_i\geq0$. As mentioned in Section \ref{sec:formulation}, both grid-scale and distributed VRE generation is included as a negative contribution to the argument $d$. This is equivalent to representing VRE as a generator with zero cost and a time-varying capacity, since zero-cost generators will always be dispatched first. The inequality constraint in \eqref{eq:power_balance} allows for negative net demand (excess VRE) to be handled at zero cost through curtailment; \eqref{eq:power_balance} is binding at optimality when $d > 0$.

\eqref{eq:economic-dispatch} admits a closed-form solution:
\begin{equation} \label{eq:closed-form-cost}
        C^*\left(d\right) = \max\left\{0,\max_{i\in\mathcal{G}}\left\{C_i(d)\right\}\right\},
\end{equation}
where
\begin{equation*}
    C_i(d) = \sum_{j:c_j \leq c_i}c_j\overline{p}_j + c_i\left(d - \sum_{j:c_j \leq c_i} \overline{p}_j\right).
\end{equation*}
The solution \eqref{eq:closed-form-cost} recovers the merit-order dispatch, that is, generators are dispatched in order of increasing bids (fuel costs) to meet the net demand.

Unfortunately, directly setting $C=C^*$ would violate Assumption \ref{assumption:smoothness} because $C^*$ is not smooth. In particular, the resulting PSG is not per-sample smooth because of the discontinuous price function ${C^{*}}'$, precluding the application of first-order equilibrium-seeking methods like \eqref{eq:sgp}. We propose two smooth proxy cost functions to facilitate first-order analysis:
\begin{align}
    C_\text{LSE}(d) & \equiv \text{LSE}_\alpha^0\left(\left(C_i(d)\right)_{i\in\mathcal{G}}\right) \label{eq:lse-cost} \\
    C_{\text{Q}}(d) & \equiv \nu_2 d^2 + \nu_1 d + \nu_0, \label{eq:quad-cost}
\end{align}
where $\alpha > 0$, $\nu_2,\nu_1 \geq 0$, $\nu_0\in\mathbb{R}$ are parameters. $C_\text{LSE}$ provides an arbitrarily tight approximation of $C^*$ as $\alpha\to0$, with larger $\alpha$ offering a more well-conditioned gradient signal at the cost of modeling fidelity. Unfortunately, $C_\text{LSE}$ does not induce an MPG in general and therefore admits no convergence guarantees under \eqref{eq:sgp}. $C_\text{Q}$, on the other hand, offers only a characteristic similarity to $C^*$: it is convex and non-decreasing over positive arguments. Unlike $C_\text{LSE}$, however, $C_\text{Q}$ induces a potential game by Theorem \ref{thm:potential}.
\begin{theorem} \label{thm:potential}
    With $C=C_\text{Q}$, $(S, \Pi)$ is an MPG with $\phi$ given by \eqref{eq:stagewise-potential} and $\tilde{\phi}\left(s_0, a\right) \equiv -C_\text{Q}\left(m(s_0) + \sum_{i\in\mathcal{D}}a_i\right) - \nu_2\sum_{i\in\mathcal{N}}\left(\sum_{j\in\mathcal{D}_i}a_j\right)^2$.
\end{theorem}
\begin{proof}
    One can verify that
    \begin{equation*}
        \begin{aligned}
        \frac{\partial}{\partial a_j}\tilde{\phi}\left(s_0,a\right) & = \frac{\partial}{\partial a_j}\tilde{R}_i\left(s_0,a\right) \\
        & = -C_\text{Q}'\left(m(s_0)+\sum_{k\in\mathcal{D}}a_k\right) - 2\nu_2\sum_{k\in\mathcal{D}_i}a_k
        \end{aligned}
    \end{equation*}
    for all $i\in\mathcal{N}$ and $j\in\mathcal{D}_i$, so \eqref{eq:general-stagewise-potential-function} holds. Applying Theorem \ref{thm:potential-structure} completes the proof.
\end{proof}

Together, Theorems \ref{thm:potential-structure} and \ref{thm:potential} imply that any unstable dynamics under the gradient flow \eqref{eq:gradient-flow} exhibited when $C=C_\text{LSE}$ are attributable to the nonlinearity of $C_\text{LSE}'$. To determine whether this nonlinearity is severe enough to prevent convergence in practice, we compare the convergence behavior of \eqref{eq:sgp} between $C=C_\text{LSE}$ and $C=C_\text{Q}$ on a benchmark market in Section \ref{sec:case-study}.

\section{CASE STUDY} \label{sec:case-study}
In this section, we implement the model from Section \ref{sec:formulation} using the IEEE RTS-GMLC test case, a modern benchmark transmission system with high penetrations of VRE and a year's worth of five-minute-resolution timeseries data for both demand and VRE output \cite{barrows2019ieee}. We run \eqref{eq:sgp} for the PSG $(S,\Pi)$ with both $C=C_\text{LSE}$ and $C=C_\text{Q}$ across ownership structures ranging from competitive to monopolistic. For reproducibility, we note that time is measured in hours, power in gigawatts, energy in gigawatt-hours, and cost in thousands of US dollars.

\subsection{Baseline demand model} \label{sec:baseline-demand}
As detailed in Section \ref{sec:formulation}, net demand includes positive contributions from consumer load as well as negative contributions from VRE. Net demand exhibits coupled annual, weekly, and daily (including diurnal solar) cycles and is characterized by shocks, often correlated with weather, that persist for multiple timesteps before dying out. As such, we propose a nested Fourier model with autoregressive noise for the net demand:
\begin{equation} \label{eq:fourier-model}
    m(s_0) = F_{M_0}^3\left(v_\text{year},s_{0,\text{year}},s_{0,\text{week}},s_{0,\text{day}}\right) + s_{0,\text{AR}},
\end{equation}
where $F_{M}^3:\mathbb{R}^{(2M+1)^3}\times\mathbb{R}\times\mathbb{R}\times{\mathbb{R}}\to\mathbb{R}$ is the order-$M$ depth-3 nested Fourier model
\begin{align*}
    F_M^3\left(v_3,s_1,s_2,s_3\right) & \equiv F_M\left(v_1,s_1\right) \\
    v_{1,i} & = F_M\left(v_{2,i},s_2\right) & \forall i \in -M, \dots, M \\
    v_{2,ij} & = F_M\left(v_{3,ij},s_3\right) & \begin{aligned}
    \forall i \in -M, \dots, M \\
    \forall j \in -M, \dots, M
    \end{aligned},
\end{align*}
and $F_M:\mathbb{R}^{2M+1}\times\mathbb{R}\to\mathbb{R}$, in turn, is the order-$M$ depth-1 Fourier model
\begin{equation*}
    F_M(v,x) \equiv v_0 + \sum_{i=1}^M\left(v_i\cos\left(2\pi i x\right) + v_{-i}\sin\left(2\pi i x\right)\right).
\end{equation*}
The fractional parts of the public market state components $s_{0,\text{year}}$, $s_{0,\text{week}}$, and $s_{0,\text{day}}$ denote the current progress through their respective seasonality, e.g., $s_{0, \text{week}}=3.7$ indicates $70\%$ progress through the week beginning at midnight on Monday. The component $s_{0, \text{AR}}$ accumulates autoregressive noise. Accordingly, the state dynamics are
\begin{equation} \label{eq:public-state-dynamics}
\begin{aligned}
    \tilde{f}_0\left(s_0,\omega\right) = \begin{bmatrix}
        \tilde{f}_{0,\text{year}}\left(s_0,\omega\right) \\
        \tilde{f}_{0,\text{week}}\left(s_0,\omega\right) \\
        \tilde{f}_{0,\text{day}}\left(s_0,\omega\right) \\
        \tilde{f}_{0,\text{AR}}\left(s_0,\omega\right)
    \end{bmatrix} = \begin{bmatrix}
        s_{0,\text{year}} + \frac{\Delta t}{T_\text{year}} \\
        s_{0,\text{week}} + \frac{\Delta t}{T_\text{week}} \\
        s_{0,\text{day}} + \frac{\Delta t}{T_\text{day}} \\
        r s_{0,\text{AR}} + \omega
    \end{bmatrix},
\end{aligned}
\end{equation}
where $T_\text{year}$, $T_\text{week}$, and $T_\text{day}$ are the lengths of their respective seasonalities, $r\in(0,1)$ is the autoregressive parameter, and $\tilde{f}_{0,\bullet}$ indicates the component of $\tilde{f}_{0}$ associated with $s_{0,\bullet}$. The transition noise is drawn from a normal distribution, $P_\omega = N(0,\sigma^2)$.

We fit the model \eqref{eq:fourier-model} to the net demand profile constructed by aggregating the load and VRE generation profiles from RTS-GMLC. The data features three distinct consumer demand profiles, one for each network area, as well as separate generation curves for each grid-scale PV, wind, impoundment hydropower, run-of-river hydropower, rooftop PV, and concentrated solar power generator.

We omit the single concentrated solar power plant. For simplicity, we assume that all impoundment hydropower plants immediately convert all inflow to power output, as opposed to storing it for later generation. The full output capacity of the single nuclear plant is considered a constant negative contribution to the net demand, since nuclear plants typically run at full capacity year-round.

In the context of the model \eqref{eq:fourier-model}, the RTS-GMLC net load profile provides the trajectory $\tau = (s_{0,\text{year},t}, s_{0,\text{week},t}, s_{0,\text{day},t},m(s_{0,t}))_{t=1}^{N_\text{data}}$. Using $M_0=5$, we fit the parameters $r$, $\sigma$, and $v_\text{year}$ to the trajectory $\tau$ using Cochrane-Orcutt estimation \cite{cochrane1949application}, which requires solving a series of least-squares problems. While one could find a (locally optimal) maximum likelihood parameterization of \eqref{eq:fourier-model} by repeating the Cochrane-Orcutt procedure until convergence, we find that a single iteration is sufficient to produce synthetic data that is nearly indistinguishable from $\tau$. Figure \ref{fig:demand} shows a single simulated demand profile generated from \eqref{eq:fourier-model}--\eqref{eq:public-state-dynamics} alongside the RTS-GMLC training data at both annual and monthly resolutions.

\begin{figure}
    \centering
    \includegraphics[width=\linewidth]{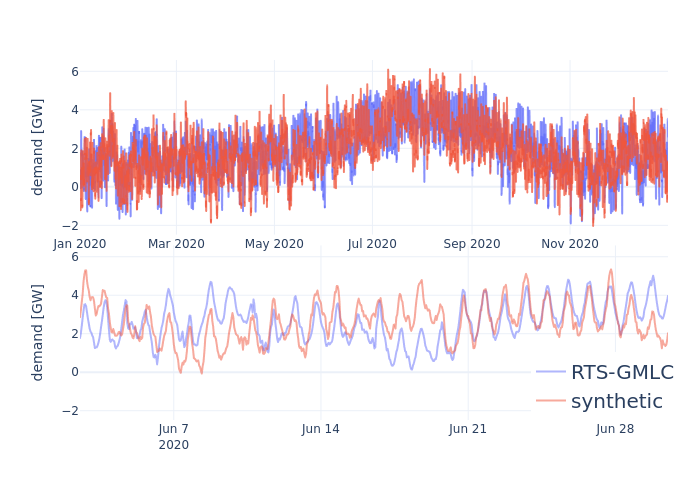}
    \caption{Synthetic load profile trained on the RTS-GMLC data, with the original data for reference.}
    \label{fig:demand}
\end{figure}

Notice that the aggregate net demand profile in Figure \ref{fig:demand} frequently drops below zero in the winter months due to excess VRE capacity. Section \ref{sec:supply-curve} details how this behavior results in large price swings with high arbitrage potential.

\subsection{Generator supply curve} \label{sec:supply-curve}
The RTS-GMLC data includes detailed models of coal, oil, and natural gas plants. Each fossil generator is specified by a fuel price, heat rate curve, and minimum/maximum generation limits. Since we neglect unit commitment, we replace all minimum generation limits with zero. We take the weighted average of the heat rate curves over the full generator range (from zero to its maximum capacity) and multiply by the fuel price to obtain a single linear cost coefficient. The resulting maximum output and cost coefficients are used to construct the expressions of $C_i$ for $i\in\mathcal{G}$.

We use smoothing parameter $\alpha=1$ in the expression of $C_\text{LSE}$. We choose the parameters $\nu$ in the expression of $C_\text{Q}$ by fitting to $C^*$ over the aggregated RTS-GMLC net demand trajectory $\tau$ using least-squares:
\begin{equation*}
    \nu = \arg\min_\nu \sum_{t=1}^{N_\text{data}}\left(C_\text{Q}\left(m\left(s_{0,t}\right)\right) - C^*\left(m\left(s_{0,t}\right)\right)\right)^2.
\end{equation*}

Figure \ref{fig:price} shows the resulting price functions ${C^*}'$, $C'_\text{LSE}$, and $C'_\text{Q}$ over the range of the RTS-GMLC net demand profile. Notice that the exact and LSE price functions are characterized by a steep jump between the negative-net-demand region, where VRE is curtailed and the price is zero, and the positive-net-demand region, where prices gradually increase from around \$20k/GWh to around \$35k/GWh. Recall from Section \ref{sec:baseline-demand} that the net demand profile frequently moves across this boundary. As such, the jump at zero demand is both an arbitrage opportunity and the dominant source of nonlinearity that might induce instability in the gradient flow dynamics \eqref{eq:gradient-flow} when $C=C_\text{LSE}$.

\begin{figure}
    \centering
    \begin{subfigure}[b]{0.23\textwidth}
        \includegraphics[width=\textwidth]{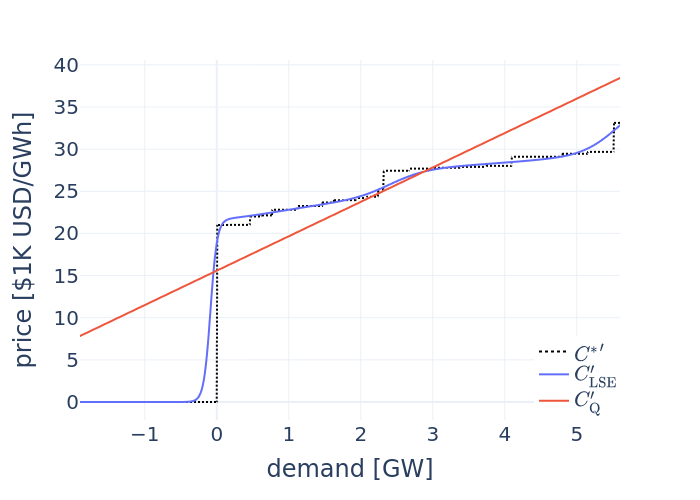}
        \caption{The LSE and quadratic price function proxies.}
        \label{fig:price}
    \end{subfigure}
    \begin{subfigure}[b]{0.23\textwidth}
        \includegraphics[width=\textwidth]{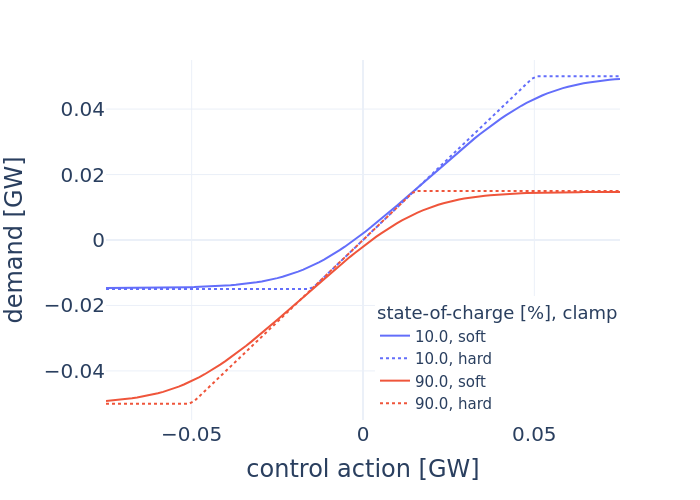}
        \caption{Actuator function at 10\% and 90\% state-of-charge.}
        \label{fig:action-map}
    \end{subfigure}
    
    \caption{Price and storage actuator functions.}
    \label{fig:smoothing-parameters}
\end{figure}

\subsection{Storage model}
The framework proposed in Section \ref{sec:formulation} generally supports grid-scale storage, load aggregators, or any other flexible self-scheduling market participants. Since the RTS-GMLC data includes storage but has no native distributed energy resources or aggregations thereof, we only consider grid-scale battery storage in this case study.

The state of battery $i\in\mathcal{D}$ is a scalar denoting its current state-of-charge, that is $\mathcal{S}_i\subset\mathbb{R}$ for all $i\in\mathcal{D}$. Battery $i$'s action space map is
\begin{equation} \label{eq:action-space}
    \tilde{\mathcal{A}}_i(s) = \left\{a_i: \begin{aligned}-\overline{a}_i & \leq a_i \leq \overline{a}_i \\ -s & \leq \Delta t a_i \leq \overline{s}_i - s\end{aligned}\right\},
\end{equation}
where $\Delta t$ is the timestep length, $\overline{a}_i$ is the maximum charge/discharge rate of battery $i$, and $\overline{s}_i$ is its maximum state-of-charge. Its state dynamics are
\begin{equation} \label{eq:storage-dynamics}
    \tilde{f}_i\left(s_0,s_i,a_i,\omega\right) = s_i + \Delta t a_i.
\end{equation}
The second constraint in \eqref{eq:action-space} implicitly enforces $0 \leq s_{i,t} \leq \overline{s}_i$ for all $t\in\mathbb{Z}_+$. For simplicity, \eqref{eq:storage-dynamics} assumes $100\%$ round-trip efficiency. Battery $i$'s policy consists of two layers: the parameterized layer $g_i:\mathcal{S}_0\times\mathcal{S}_i \to \mathbb{R}$ that maps the current state to an unconstrained control action and an actuator layer $h_i: \mathbb{R} \to \tilde{\mathcal{A}}_i^\cup$ that maps the control action to the feasible action space $\tilde{\mathcal{A}}_i(s_i)$. We design the actuator layer as
\begin{equation*}
    \begin{aligned}
    h_i\left(a,s\right) = a & - \text{LSE}_\beta^0\left(a + \text{LSE}_\beta\left(-\overline{a}_i,\frac{s - \overline{s}_i}{\Delta t}\right)\right) \\
    & + \text{LSE}_\beta^0\left(\text{LSE}_\beta\left(-\overline{a}_i,-\frac{s}{\Delta t}\right) - a\right).
    \end{aligned}
\end{equation*}
$h_i$ is designed as a smooth approximation to a clamping function that preserves a locally linear region around $a = 0$. We use smoothing parameter $\beta=0.01$. Figure \ref{fig:action-map} shows $h_i$ for one of the RTS-GMLC batteries (since the batteries are identical, all $h_i$ are identical as well), evaluated at $s_i = 0.1\overline{s}_i$ and $s_i = 0.9\overline{s}_i$. The approximated hard-clamp function is also shown as a dashed line for reference.

We propose the parameterized layer
\begin{equation}  \label{eq:parameterized-layer}
    \begin{aligned}
    g_i\left(s_0,s_i|\theta_i\right) & = F_{M_i}^3\left(\theta_F,s_{0,\text{year}},s_{0,\text{week}},s_{0,\text{day}}\right) \\
    & + \theta_{\text{AR}}s_{0,\text{AR}} + \theta_{\text{SOC}}s_i.
    \end{aligned}
\end{equation}

Intuitively, the parameterized layer of device $i$'s policy is a simple linear function of the current residual, device $i$'s state-of-charge, and the Fourier basis introduced in \eqref{eq:fourier-model} up to order $M_i$. Inclusion of the Fourier basis as a feature allows the linear policy to construct an internal representation of the baseline demand. $g_i$ allows for basic arbitrage and withholding behavior if agents learn weights that modulate this baseline demand model; see Section \ref{sec:results}/Figure \ref{fig:rollout} for an example. While deep neural networks can encode more granular policies, the linear parameterization $g_i$ is easier to train, more interpretable, and sufficiently expressive for the purposes of studying the fundamental dynamics of the market.

The full policy for device $i\in\mathcal{D}$ is
\begin{equation*}
    \tilde{\mu}_i\left(s_0,s_i|\theta_i\right) = h_i\left(g_i\left(s_0,s_i|\theta_i\right),s_i\right).
\end{equation*}

RTS-GMLC features a single grid-scale battery storage unit with a three-hour duration and a maximum net load of 50 MW. To represent a market with a much larger storage penetration, we include 48 copies of this battery for a total installed capacity of 2.4 GW (7.2 GWh), approximating the typical amplitude of the baseline demand; see Figure \ref{fig:demand}. We set $M_i = 2$ in the parameterized layer \eqref{eq:parameterized-layer}, giving a total of 128 policy parameters per device.

\subsection{Training procedure and hyperparameters} \label{sec:hyperparameters}
For each variant $C=C_\text{LSE}$ and $C=C_\text{Q}$, we begin by training a ``social planner'' policy by initializing $\theta = 0$ and running \eqref{eq:sgp} (implemented with Adam \cite{kingma2014adam} throughout for reliable training) until convergence on the modified game with $R_i = -C$. Intuitively, the social planner policies maximize social welfare by minimizing the total (discounted) fuel cost required to serve the baseline load profile. \eqref{eq:sgp} in this context is equivalent to stochastic gradient ascent on a single payoff function (cooperative games are a special case of potential games). We denote the social-planner return function analogous to \eqref{eq:expectea_return} as $J_0$.

Recall that the storage fleet consists of 48 identical batteries. We consider ten symmetric ownership structures associated with the ten factors of 48: a single agent owning all 48 batteries, two agents owning 24 each, and so on up to 48 agents owning one battery each. We refer to the fraction of the storage fleet owned by each agent, $1/|\mathcal{N}|$, as its \textit{market share}; the ten structures thus correspond to per-agent market shares of 100\%, 50\%, 33.3\%, and so on down to 2.1\%, as labeled in Figures \ref{fig:training}--\ref{fig:returns}. Ownership enters the model only through the payoff \eqref{eq:payout}, which compensates each agent for the aggregate output of its fleet. Since the physical fleet is identical across structures, differences in equilibrium outcomes are attributable entirely to market power. For each ownership structure, we initialize the agent policies to the trained social planner policy and run \eqref{eq:sgp}. Since both the initialization and the update rule \eqref{eq:sgp} are symmetric across agents, all agents' policies remain identical throughout training. Consequently, $J_i(\theta^{(k)})=J_j(\theta^{(k)})$ for all $i,j\in\mathcal{N}$ at every iteration $k$, and the total storage revenue $\sum_{i\in\mathcal{N}}J_i(\theta)$ reduces to $|\mathcal{N}|$ times any single agent's return.

We use a timestep $\Delta t=1$ and downsample the RTS-GMLC profiles accordingly before fitting the baseline demand model. We choose $\gamma=0.999$, giving an effective horizon of 1000 hours (about 42 days). For hardware efficiency, we truncate rollouts at $T = 5000$ ($0.999^{5000} < 0.01$) rather than sampling $T$ geometrically. We train social planner policies for $1000$ iterations and general-sum policies for $500$ iterations each. We use batch size $B = 1000$ and learning rate $\eta^{(k)} = 0.001$. The initial state distribution $P_0$ is a uniform distribution over the set $\{s:s_{0,\text{AR}}=0,0 \leq s_{0,\text{year}},s_{0,\text{week}},s_{0,\text{day}} \leq 1, s_i=\frac{\overline{s}_i}{2}\quad\forall i \in \mathcal{D}\}$.

\subsection{Results} \label{sec:results}
Policies exhibited stable convergence under \eqref{eq:sgp} for all ten ownership configurations with $C=C_\text{Q}$ as well as both the social planner configurations. Since convergence is guaranteed for these potential-game configurations, their training curves are not shown here. The top subplot of Figure \ref{fig:training} visualizes the first 300 iterations of \eqref{eq:sgp} for the ownership configurations where convergence is not guaranteed, that is, when $C=C_\text{LSE}$. ``Total storage revenue'' on the vertical axis refers to $\sum_{i\in\mathcal{N}}J_i(\theta)$ (recall from Section \ref{sec:hyperparameters} that this equals $|\mathcal{N}|$ times the per-agent return by symmetry). The bottom subplot shows the norms of the parameter steps $\|\theta^{(k+1)} - \theta^{(k)}\|$. The training trajectories are consistent with convergence to a neighborhood of an LASE: agents take larger parameter steps that increase their returns as they approach the LASE, then take smaller, noisier steps within the neighborhood as the expected gradient vanishes.

\begin{figure}
    \centering
    \includegraphics[width=1.0\linewidth]{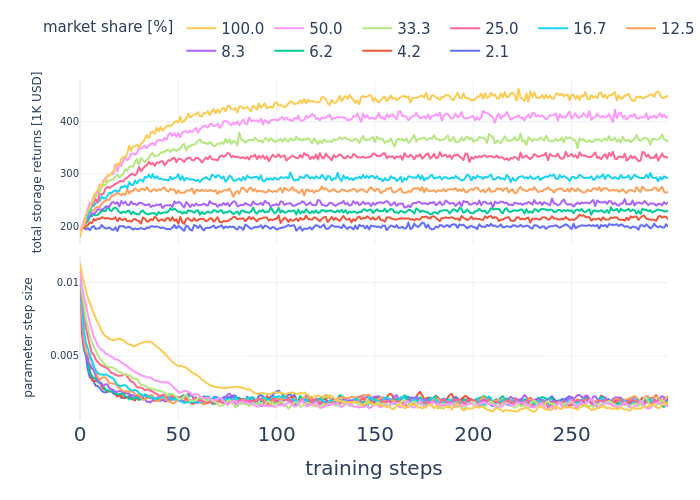}
    \caption{LSE training curves.}
    \label{fig:training}
\end{figure}

\begin{figure}
    \centering
    \includegraphics[width=1.0\linewidth]{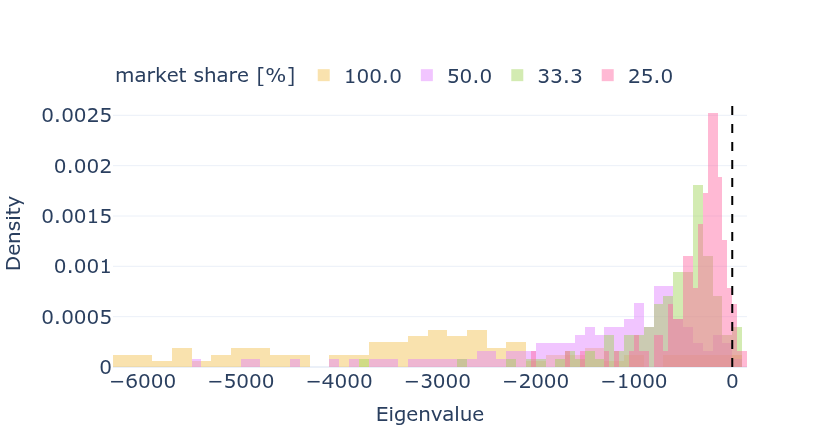}
    \caption{Empirical Hessian eigenvalue distribution for the trained LSE policies.}
    \label{fig:LSE_eigval}
\end{figure}

To determine whether the LASEs found for the $C=C_\text{LSE}$ configurations are differential Nash equilibria, we check the eigenvalues of the empirical single-agent Hessians $\nabla_{\theta_i}^2 J_i(\theta^*)$ averaged over 100 rollouts. The eigenvalue distributions of the four configurations with the highest market shares are shown in Figure \ref{fig:LSE_eigval}. For high-market-share configurations, the eigenvalue distributions generally indicate negative definiteness. While there are a small number of slightly positive eigenvalues, these also appear in the distributions for the $C=C_\text{Q}$ configuration, even though all LASEs are known to be differential Nash equilibria for potential games. The positive eigenvalues, therefore, are likely the result of noise. The distributions for lower-market-share configurations (not shown due to space constraints) were less conclusive, with substantial portions of the distribution in the positive region. Oddly, this was true when $C=C_\text{Q}$ as well, so it is possible that the return landscape for these configurations is simply not amenable to low-variance Hessian estimates. Therefore, we conclude that the policies in high-market-share cases appear to be differential Nash equilibria, but are unable to make the same claim for the low-market-share cases. We note that the behavior of the equilibria is otherwise consistent across market shares in all other analyses conducted here.

Figure \ref{fig:returns} evaluates the trained policies for all configurations for both $C=C_\text{LSE}$ and $C=C_\text{Q}$ using two metrics: the total storage revenue $\sum_{i\in\mathcal{N}}J_i(\theta^*)$ and the ``arbitrage surplus'' $J_0(\theta^*) - J_0(0)$, each averaged across 1000 rollouts. The storage revenue is the same quantity from the vertical axis of Figure \ref{fig:training}; this is (a scaled version of) the symmetric agent-by-agent returns on which the policies were trained. The arbitrage surplus is the negative discounted total fuel cost $J_0(\theta^*)$ shifted by its baseline value (when batteries idle at each timestep) $J_0(0)$ to capture the added social value associated with each policy. The social planner policies are also evaluated according to each metric and are indicated by faded horizontal lines.

\begin{figure}
    \centering
    \includegraphics[width=1.0\linewidth]{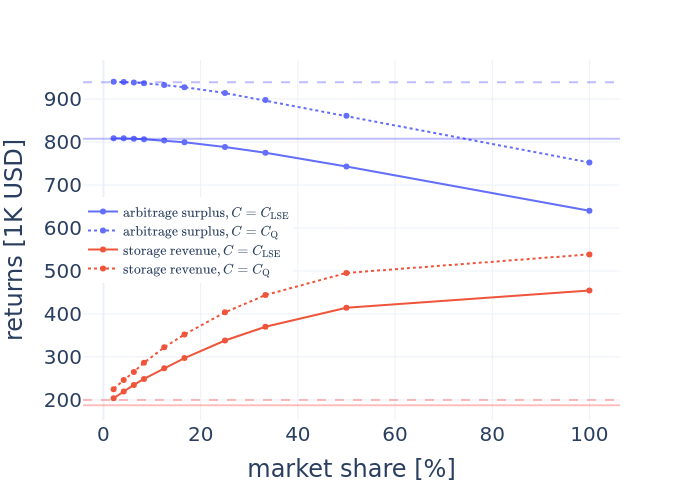}
    \caption{Arbitrage surplus (blue) and storage revenue (red) of the trained equilibrium policies versus per-agent market share, for $C=C_\text{LSE}$ (solid) and $C=C_\text{Q}$ (dotted). Faded horizontal lines mark the social-planner value of each metric.}
    \label{fig:returns}
\end{figure}

Figure \ref{fig:returns} features an interpretable trend: the monopolist maximizes storage profits at the expense of social arbitrage surplus, while the social planner maximizes arbitrage surplus at the expense of storage revenue. As management of the batteries is distributed between more agents, they each individually see less opportunity to influence the price, and both metrics approach their social-planner values (the faded lines) as market share decreases. In this context, the social planner admits an interpretation as a perfectly competitive market with an infinitely small market share per agent. The trends in Figure \ref{fig:returns} are nearly identical for $C=C_\text{LSE}$ and $C=C_\text{Q}$.

Figure \ref{fig:rollout} demonstrates a selected week from a randomly seeded year-long rollout of the social-planner and monopolist (100\% market share) policies for $C=C_\text{LSE}$. The three subplots show demand, price, and battery state-of-charge, respectively. The dotted line is the baseline case where batteries idle at all timesteps, shown for reference. The social planner dispatches the batteries aggressively to flatten out the net demand curve. The monopolist engages in qualitatively similar behavior but withholds some flexibility by charging and discharging at lower rates, thereby maintaining a higher price variance. This is a better outcome for the storage operators, but results in a higher total fuel cost (lower economic efficiency). The trend in Figure \ref{fig:returns} suggests that the equilibrium policies found by \eqref{eq:sgp} gradually shift between the social-planner and monopolist behaviors exhibited in Figure \ref{fig:rollout} as per-agent market share increases. Rollouts of the intermediate configurations (not shown) are consistent with this hypothesis, exhibiting progressively lower charge and discharge rates as market share increases.

\begin{figure}
    \centering
    \includegraphics[width=1.0\linewidth]{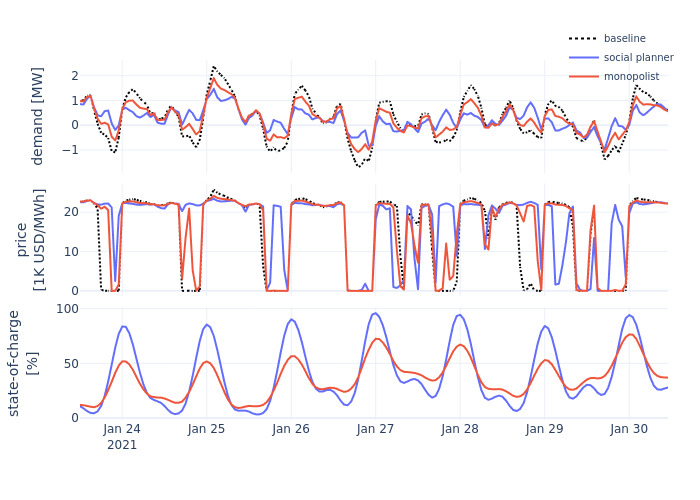}
    \caption{Rollout of the social-planner and monopolist LSE policies.}
    \label{fig:rollout}
\end{figure}

Overall, the update rule \eqref{eq:sgp} appears to exhibit almost identical performance for $C=C_\text{LSE}$ and $C=C_\text{Q}$. In particular, Figure \ref{fig:training} indicates that \eqref{eq:sgp} finds LASEs for $C=C_\text{LSE}$ across all per-agent market share configurations. Figure \ref{fig:LSE_eigval} suggests that these LASEs also happen to be Nash equilibria for at least the configurations with high per-agent market shares, indicating that the nonlinearities in $C'_\text{LSE}$ did not critically distort the vector field in \eqref{eq:gradient-flow} relative to the conservative field induced by $C_\text{Q}$.

\section{CONCLUSION}
We propose a generalized multi-agent policy optimization model for characterizing strategic equilibria in real-time electricity markets with large, arbitrage-based participants such as grid-scale storage and load aggregators. Under appropriate cost function proxies, the model admits natural gradient-based equilibrium-seeking update rules in the mold of \eqref{eq:gradient-flow}. We study two proxies: $C_\text{Q}$, which is stylized but induces a Markov potential game, and $C_\text{LSE}$, which is high-fidelity but admits no convergence guarantees. We find that the standard gradient play update rule \eqref{eq:sgp} converges to similar equilibria for both cost functions on a standard benchmark market, and that these equilibria are useful for characterizing arbitrage withholding and quantifying the associated economic inefficiencies.

Generalizing the model to incorporate transmission constraints is a natural extension of this work. This would preclude a closed-form solution in the form of \eqref{eq:closed-form-cost}; a smooth proxy analogous to \eqref{eq:lse-cost} can instead be obtained by replacing the hard constraints of the market-clearing optimal power flow problem with log barriers and computing prices via differentiable convex optimization \cite{agrawal2019differentiable}. Additionally, we informally assume that differential Nash equilibria under $C=C_\text{LSE}$ with small $\alpha$ approximate local Nash equilibria under $C=C^*$; formalizing this relationship is future work.

\bibliographystyle{ieeetr}
\bibliography{references}

\end{document}